\documentclass[journal]{IEEEtran}
\usepackage{indentfirst,amsmath}
\usepackage{color,rotating,pdflscape}
\usepackage{amsfonts,amsmath,amssymb,amsthm}
\usepackage{latexsym,amscd}
\usepackage{amsbsy,extarrows}
\usepackage{graphicx,multirow,bm}
\usepackage[table]{xcolor}
\usepackage{ulem}
\usepackage{extarrows,cite}
\usepackage{tikz}
\usetikzlibrary{arrows,shapes}
\usetikzlibrary{shadows}

\newtheorem{Theorem}{Theorem}
\newtheorem{Corollary}{Corollary}

\newtheorem{Example}{Example}

\newtheorem{Remark}{Remark}

\begin{document}

\title{A Generic Transformation to Enable  Optimal Repair in MDS Codes for Distributed Storage Systems}
\author{Jie Li, \IEEEmembership{Member,~IEEE}, Xiaohu Tang, \IEEEmembership{Member,~IEEE}, and Chao Tian, \IEEEmembership{Senior Member,~IEEE}
\thanks{The work of J. Li and X. Tang was supported in part by the National Science Foundation of China under Grant 61325005, the Major Frontier Project of Sichuan Province under Grant 2015JY0282, and  the Project funded by China Postdoctoral Science Foundation under Grant 2017M622385.
The work of C. Tian was supported in part by the National Science Foundation under Grants CCF-15-26095 and CCF-18-32309.
This paper was  presented in part at the 2017 IEEE International Symposium on
Information Theory, Aachen, Germany, and in part at the 2018 IEEE International Symposium on
Information Theory, Vail, CO, USA.}
\thanks{J. Li was with the Information Security and National Computing Grid Laboratory, Southwest Jiaotong University, Chengdu, 610031, China. He is now with the Hubei Key Laboratory of Applied Mathematics, Faculty of Mathematics and Statistics, Hubei University,
Wuhan 430062, China (e-mail: jieli873@gmail.com).}
\thanks{X.H. Tang is with the Information Security and National Computing Grid Laboratory, Southwest Jiaotong University, Chengdu, 610031, China (e-mail: xhutang@swjtu.edu.cn).}
\thanks{C. Tian was with the Department of Electrical Engineering and Computer Science at the University of Tennessee Knoxville. He is now with the Department of Electrical and Computer Engineering, Texas A\&M University, College Station, TX 77843, USA (e-mail: chao.tian@tamu.edu).}
}
\date{}
\maketitle

\begin{abstract}
We propose a generic  transformation that can convert any nonbinary   $(n=k+r,k)$ maximum distance separable (MDS) code into another  $(n,k)$ MDS code over the same field such that 1) some arbitrarily chosen $r$ nodes have the optimal repair bandwidth and the optimal rebuilding access, 2) for the remaining $k$ nodes, the normalized repair bandwidth and the normalized rebuilding access (over the file size) are preserved, 3) the sub-packetization level is increased only by a factor of $r$.
Two immediate applications of this generic transformation are then presented.  The first application is that we can transform any nonbinary MDS code with the optimal repair bandwidth or the optimal rebuilding access for the systematic nodes only, into a new MDS code which possesses the corresponding repair optimality for all nodes. The second application is that by applying the transformation multiple times, any nonbinary   $(n,k)$ scalar MDS code can be converted into an $(n,k)$ MDS code with the optimal repair bandwidth and the optimal rebuilding access for all nodes, or only a subset of nodes, whose sub-packetization level is also optimal.
\end{abstract}


\section{Introduction}

Distributed storage systems built on a large number of unreliable storage nodes have important
applications in large-scale data center settings, such as Facebook's coded Hadoop, Google Colossus, and Microsoft Azure \cite{Micro}, and in peer-to-peer storage settings,
such as OceanStore \cite{ocean}, Total Recall \cite{total}, and DHash++ \cite{Dhash}. To ensure reliability, redundancy is imperative for these systems. Generally speaking,
there are two mechanisms to introduce redundancy, namely replication and erasure coding. Comparing with the former, erasure coding can provide
higher reliability at the same redundancy level, and thus is more attractive.

When a storage node fails, a self-sustaining distributed storage system should make a repair to maintain the continuing operation of the overall system. During the repair process, the \textit{repair bandwidth}, which is defined as the amount of data downloaded from the surviving nodes to repair the failed node, should be minimized. The repair bandwidth of the classic MDS erasure codes, such as Reed-Solomon codes \cite{RS-codes}, is rather excessive because they rely on a naive repair strategy, \textit{i.e.},  to first reconstruct the original file, and then repair the failed node.

The repair problem was first brought into the spotlight by Dimakis \textit{et al.} \cite{Dimakis}.   As a result, the optimal repair bandwidth and the optimal rebuilding access\footnote{The seminal work in \cite{Dimakis} identified two different repair modes, namely, exact repair and functional repair. Under exact repair, a replacement node is required to store exactly the same data as that was stored in the failed node; in contrast, under functional repair, a failed node is replaced by a node that is functionally equivalent. In this paper we consider the first case which is practically more important.} were subsequently established \cite{Dimakis,Barg1}.
A node of an $(n,k)$ MDS code with a sub-packetization level $N$  is said to have the optimal repair bandwidth if the repair bandwidth is $\gamma^*(d)\triangleq\frac{d}{(d-k+1)}N$, and is said to have  the optimal rebuilding access  if the amount of data accessed is also $\gamma^*(d)$, where $d$ ($k\le d\le n-1$) is the number of surviving nodes accessed during the repair process. Various explicit or less explicit code constructions have been proposed in the literature, usually for certain restricted parameter ranges, where some of the notable  works are \cite{Goparaju,product,Suh,repair-parity-zigzag,hadamard,Sasidharan-Kumar,Zigzag,Hadamard-strategy,extend-zigzag}.
Specifically, most of the aforementioned works \cite{Suh,repair-parity-zigzag,hadamard,Sasidharan-Kumar,Zigzag,Hadamard-strategy,extend-zigzag} consider the case $d=n-1$ to maximally reduce the repair bandwidth, since the minimum repair bandwidth $\gamma^*(d)$ is a decreasing function of $d$; this setting is also the focus of this work.

The initial motivation of our work is the following observation. At the practically more important range of high-rate case, \textit{i.e.}, $k/n > 1/2$, most early constructions that are able to optimally repair any single node failure are limited and usually restricted to a small number of parity nodes \cite{repair-parity-zigzag,hadamard,Sasidharan-Kumar,Zigzag,Hadamard-strategy,extend-zigzag}. In contrast, there exist more code constructions that can  optimally repair any failure of the systematic nodes \cite{Goparaju,hadamard,Zigzag,invariant-subspace,Etzion,Long-IT} with less restrictions on the parameters. This phenomenon left the impression that the latter is considerably simpler than the former, which intrigued us to seek better understanding of this perceived barrier.

Our quest eventually led to a very powerful transformation, which is
the subject of this paper. More precisely, we provide a transformation
that can convert any  nonbinary   MDS  code into another MDS
code, which endows any $r=n-k$ chosen nodes with the optimal repair bandwidth and the optimal rebuilding access properties, and at the same time, preserves the normalized  repair bandwidth and the normalized rebuilding
access for the remaining $k$ nodes. The resultant code uses the same finite field as the base
code, and has a sub-packetization level a factor of $r$ larger. As two
immediate applications of this transformation, we show that
1) any nonbinary   MDS  code with the optimal repair bandwidth
or the optimal rebuilding access for the systematic nodes  only can
be converted into an MDS  code with the corresponding repair
optimality for all nodes, and 2)  by applying the transformation multiple times, any nonbinary   $(n,k)$ scalar  MDS code can be converted into an $(n,k)$ MDS  code with the optimal repair bandwidth and  the optimal rebuilding access for all nodes (or a desired subset of the nodes). In the second application, the  resultant codes have the optimal  sub-packetization level, which matches the lower bounds recently identified in \cite{tight-bound-on-alpha} for $(n,k)$ MDS  codes with the optimal rebuilding access.

The remainder of the paper is organized as follows. Section
II gives  some historical notes and explains the relation to several existing works. Section III presents some necessary preliminaries. The generic transformation is given in Section IV, followed by the proofs of
the asserted properties. Two important applications of this transformation are discussed in Section V. Finally, Section VI provides some concluding remarks.

\section{Historical Notes and Relation to Existing Works}
\label{sec:history}

As explained in the early version \cite{transform-arxiv} of this paper, we were initially motivated to seek an explanation of the perceived technical barrier, and to provide a construction of high-rate MDS  codes that can optimally repair all nodes, based on existing MDS  codes that can only optimally repair the systematic nodes. Independent and parallel to our work, Ye and Barg \cite{Barg1,Barg2} proposed several explicit constructions of high-rate MDS  codes that can optimally repair all nodes. Particularly, the codes in \cite{Barg1} allow the number of helper nodes to be anywhere from $k+1$ to $n-1$, and they also allow simultaneous repair of multiple node failures, solving the problem of constructing MDS codes with the optimal repair bandwidth in full generality. Moreover, the code in \cite{Barg2} has the optimal sub-packetization level with respect to the lower bound for $(n,k)$ MDS  codes with the optimal rebuilding access given in \cite{tight-bound-on-alpha}. Shortly after, Sasidharan \textit{et al.} \cite{Sasidharan-Kumar2,Sasidharan-Kumar3} independently discovered two code constructions based on a neat data cube representation. The construction in \cite{Barg2} and that in \cite{Sasidharan-Kumar2} turn out to be essentially equivalent. One key new ingredient in \cite{Barg1,Barg2,Sasidharan-Kumar2,Sasidharan-Kumar3},  in contrast to most previous efforts, is that these constructions are given in terms of parity-check matrix, and as a consequence they do not distinguish between the systematic nodes and the parity nodes at all. The success of these constructions can essentially be interpreted as also showing  that the aforementioned barrier is only a misimpression, however without directly addressing the relation between the two different repair requirements. After our initial discovery of the transformation \cite{transform-arxiv} and during the preparation of \cite{transform-ISIT}, it became clear to us that this transformation is more powerful than we had originally realized, which led to the current form of presentation as a generic code transformation and its applications; see also \cite{tight-bound-on-alpha} for a discussion on these closely related discoveries.

\begin{table*}[htbp]
\begin{center}
\caption{A comparison between the $(n,k)$ piggyback codes in \cite{piggyback-arxiv,Yang-piggyback} and the resultant $(n,k)$ MDS codes obtained from the first application in Section IV
}\label{Table comp}
\begin{tabular}{|c|c|c|c|c|}
\hline
&\multirow{2}{*}{Sub-packetization level $N$}& \multirow{2}{*}{Field size } &The ratio of repair bandwidth for  & \multirow{2}{*}{Remark} \\
&   &  &the parity nodes to the optimal value  & \\
\hline\hline
Base code  & $N'$ &  $q$ & $\frac{nr-r^2}{n-1}$ & not optimal \\
\hline
  Piggyback code in \cite{piggyback-arxiv}               &$2N'$   &  $q$  &  $>\frac{1}{2}\frac{nr-r^2}{n-1}$ &not optimal    \\
\hline
 Piggyback code in \cite{Yang-piggyback} &   $rN'$  & $q$   & $\frac{n+r^2-2r}{n-1}$  &asymptotically optimal \\
\hline
Codes obtained from the first application  & $rN'$ &  $q(q\ge 3)$ & $1$ &optimal \\
\hline
\end{tabular}
\end{center}
\end{table*}

\begin{table*}[htbp]
\begin{center}
\caption{A comparison of some parameters between the  $(n,k)$ MDS  codes in \cite{Barg1,Barg2} and the explicit $(n,k)$ MDS  codes obtained by applying the transformation once  to the Hadamard design code \cite{hadamard}, the zigzag code \cite{Zigzag}, and the optimal access code together with the long MDS code \cite{Long-IT,invariant-subspace}}\label{Table comp Ye}
\begin{tabular}{|c|c|c|c|}
\hline
               & Sub-packetization level $N$ &Field size $q$ &  Remark  \\
\hline\hline
Ye-Barg code 1 \cite{Barg1}  &  $r^{n}$ &   $q \ge rn$ &\\
\hline
The first application on  & \multirow{2}{*}{$r^{k+1}$}& \multirow{2}{*}{$q\ge 2k+1$  and $q$ is odd if $r=2$ } &   \\
Hadmard design code \cite{hadamard}&  & & \\
\hline
\hline
Ye-Barg code 2 \cite{Barg1}  &  $r^{n-1}$  &   $q>n$ &\\
\hline
The first application  & \multirow{2}{*}{$r^k$} & $q=3$ if $r=2$&\\
on Zigzag code \cite{Zigzag}  &   & $q=4$ if $r=3$ &  \\
\hline
\hline
Ye-Barg code 3 \cite{Barg2}  &  $r^{\lceil\frac{n}{r}\rceil}$  &    $q\ge r\lceil{n\over r}\rceil$&  \\
\hline
The first application on & \multirow{2}{*}{$r^{\frac{n}{r}}$}& $q>k/2$ if $r=2$  &\multirow{2}{*}{$r|n$} \\
 optimal access code  \cite{Etzion}&   &$q>k$ and $q$ is even if $r=3$ &\\
\hline
The first application on  & \multirow{2}{*}{$r^{\frac{n+1}{r+1}}$}& \multirow{2}{*}{$q>2k$ if $r=2$} &\multirow{2}{*}{$(r+1)|(n+1)$}  \\
Long MDS code \cite{Long-IT,invariant-subspace}&   && \\
\hline
\end{tabular}
\end{center}
\end{table*}

\begin{table*}[htbp]
\begin{center}
\caption{A comparison between the $(n,k)$ MDS  codes proposed in \cite{Barg2,Sasidharan-Kumar2}  and the $(n,k)$ codes obtained from the second application in Section IV
}\label{Table comp Ye Kumar}
\begin{tabular}{|c|c|c|c|c|}
\hline
&{Sub-packetization level $N$}& {Field size }  &{Optimal rebuilding access for all nodes} \\
\hline\hline
Ye-Barg code 3  \cite{Barg2}               &$r^{\lceil{n\over r}\rceil}$   &  $q\ge r{\lceil{n\over r}\rceil}$  &  Yes    \\
\hline
The PCT code  \cite{Sasidharan-Kumar2} &   $r^{\lceil{n\over r}\rceil}$  & $q\ge r{\lceil{n\over r}\rceil}$      &Yes \\
\hline
Codes obtained from the second application  & $r^{\lceil{n\over r}\rceil}$ &  $q\ge n$ & Yes \\
\hline
\end{tabular}
\end{center}
\end{table*}

In retrospect, the constructions in \cite{Barg2} and \cite{Sasidharan-Kumar2} in fact share the same core technique as ours,  which is referred to as pairwise coupling transformation (PCT) in \cite{tight-bound-on-alpha}. Therefore, we refer the MDS code constructed in \cite{Sasidharan-Kumar2} as the PCT code hereafter. The key conceptual difference is that this technique was presented in \cite{Barg2} and \cite{Sasidharan-Kumar2} from the perspective of the parity-check matrix while ours is from the perspective of the generator matrix. Furthermore, in \cite{Barg2} and \cite{Sasidharan-Kumar2} this core technique was applied on all the pairs simultaneously which makes the process much less explicit, whereas we isolate the pairs which helps to untangle the complicated process. As a consequence of the abstraction as a generic transformation, we can elucidate the requirements on the base MDS code, the sufficient conditions for the various components of the transformation, and the properties of the resultant code. These conditions allow us more design choices in constructing the codes, and indeed reveal coding techniques that are not possible in either \cite{Barg2} or \cite{Sasidharan-Kumar2}; see Tables
\ref{Table comp Ye} and \ref{Table comp Ye Kumar}, and Remark \ref{Rm_compare}.

One  important subtlety is that our generic transformation is based on transforming known MDS  codes. As such, if the base code is explicit, the resultant code is also explicit; however, if the base code is not explicit, then the resultant code is also not explicit. This is not a cause for concern in the second application of the transformation, since the base code is any scalar MDS code, for which well-known construction techniques exist, however, more caution is warranted in the first application where the systematic nodes in the base code need to have the optimal repair property. Particularly, when $r>3$, the code constructions in \cite{Sasidharan-Kumar,hadamard,Hadamard-strategy,Zigzag,extend-zigzag,Long-IT} are only shown to exist in a sufficiently large alphabet guaranteed by either the Schwartz-Zippel lemma or the Combinatorial Nullstellensatz \cite{Alon}. To find exact code will necessitate a search for valid assignments to the entries of the generator matrix, which may not be trivial in general. In contrast, the constructions in \cite{Sasidharan-Kumar2,Sasidharan-Kumar3,Barg1,Barg2}, those in \cite{Sasidharan-Kumar,hadamard,Hadamard-strategy,Long-IT} for $r=2$, and those in \cite{Zigzag,extend-zigzag,Etzion} for $r=2$ and $r=3$,  are explicit in the sense that the entries of the generator matrix can be assigned without any search.

We note that another thread of efforts particularly relevant to our work is the piggybacking design framework \cite{piggyback-arxiv,Hitchhiker},  which was proposed to reduce the repair bandwidth or reduce the repair-locality of a base MDS code.
We were indeed partially motivated by this design framework.
The transformation we propose here has a similar flavor as the piggybacking design, \textit{i.e.}, by operating on multiple instances\footnote{The term ``instance'' refers to a codeword obtained by applying the coding operation on part of the raw data, and our construction involves applying the same coding operation on non-overlapping parts of the raw data to obtain multiple coding instances. We adopted this terminology here to be consistent with that used in the piggybacking framework \cite{piggyback-arxiv}.} of a base code. However, the resultant code does not belong to the piggybacking design framework, since the latter stipulates that only a function of the symbols in the previous instances can be added to the symbol in the current instance, whereas our transformation does not observe this sequential order. Furthermore,  the existing piggybacking designs in \cite{piggyback-arxiv,Yang-piggyback} both suffer a loss of optimality in terms of repair bandwidth and rebuilding access.

A comparison between the piggyback codes in \cite{piggyback-arxiv,Yang-piggyback} and the resultant MDS codes obtained from the first application in Section IV is provided in Table \ref{Table comp},  a comparison between the MDS  codes proposed by Ye and Barg and the codes obtained from the first application in Section IV is provided in Table \ref{Table comp Ye},   and a comparison between the MDS  codes proposed in \cite{Barg2,Sasidharan-Kumar2}  and the codes obtained from the second application in Section IV is provided in Table \ref{Table comp Ye Kumar}. It is seen from these comparisons that   the resultant codes obtained from the generic transformation have two main advantages: 1) the optimal repair bandwidth for the parity nodes, whereas the repair bandwidth for the parity nodes of the piggyback code in \cite{piggyback-arxiv} (resp. in \cite{Yang-piggyback}) is far from optimality (resp. asymptotically optimal); 2) a lower sub-packetization level and/or a smaller field size in some cases compared to the MDS  codes in \cite{Barg1,Barg2,Sasidharan-Kumar2}.

\section{Preliminaries}
For any two integers $i<j$, denote by $[i,j]=\{i,i+1,\cdots,j\}$ and $[i,j)=\{i,i+1,\cdots,j-1\}$. Let $q$ be a prime power and $\mathbb{F}_q$ be the finite field with $q$ elements.
Assuming that a source data file comprising of $M=kN$ symbols over a finite field $\mathbb{F}_q$ is encoded by a base  $(n,k)$  MDS code, and then dispersed
across $n$ storage nodes, each  storing  $N$ symbols. In practice, a  code in systematic form is more preferred. In the systematic form, the first $k$ nodes storing the original file are named \textit{systematic nodes},
whose contents are denoted as  $\mathbf{f}_0, \mathbf{f}_1, \cdots, \mathbf{f}_{k-1}$, respectively, where $\mathbf{f}_i$ is a column vector of length $N$; the remaining nodes
are referred to as \textit{parity nodes}, whose contents are linear combinations of the data in the systematic nodes, \textit{i.e.},
 $\mathbf{f}_{k+i}=A_{i,0}\mathbf{f}_0+\cdots + A_{i,k-1}\mathbf{f}_{k-1}$, for $i\in [0,r)$, where $r=n-k$ and $A_{i,j}$ ($j \in [0,k)$) is
an $N\times N$ matrix over $\mathbb{F}_q$, termed  the \textit{coding matrix}  of systematic  node $j$ for parity node $i$. Systematic node $j$ and parity node $i$ are also respectively termed node $j$ and node $k+i$ for convenience. Note that an MDS code is also called a scalar MDS code if $N=1$ and a vector MDS code if $N>1$. The structure of an  $(n,k)$ systematic  MDS code can be
specified by the following equations,
\begin{eqnarray*}\label{MSR_Model}
\mathbf{f}_{k+i}=A_{i,0}\mathbf{f}_0+\cdots + A_{i,k-1}\mathbf{f}_{k-1},~~i\in[0,r).
\end{eqnarray*}
An $(n,k)$ MDS code  has \textit{the MDS property} that the source data file can be reconstructed by connecting any $k$ out of the $n$ nodes, and is preferred  to have \textit{the optimal repair bandwidth}, \textit{i.e.},  any failed node $i$ can be repaired by downloading $N/ r$ symbols from each surviving node $j$,  $j\in [0,n)\backslash\{i\}$. In addition to the optimal repair bandwidth, it is also desirable if the nodes have the \textit{optimal rebuilding access}. That is, when repairing a failed node, only $N/r$ symbols are accessed at each surviving node, \textit{i.e.}, the minimum amount of data is accessed at each surviving node \cite{Barg2}. This appealing property enhances the repair bandwidth requirement, and codes with this property are capable of substantially reducing the disk I/O overhead during the repair process.

For a general MDS code with or without special repair ability, we associate with each   node
$i$  a   repair bandwidth profile
\begin{eqnarray*}
 \bm{\beta}_i\triangleq (\beta_{i,0},\beta_{i,1},\cdots,\beta_{i,i-1},\beta_{i,i+1},\cdots,\beta_{i,n-1}),
\end{eqnarray*}
where $\beta_{i,j}$  denotes the amount of symbols sent from node $j$  when repairing node $i$. The data sent from node $j$  when repairing node $i$ is normally obtained by multiplying $\mathbf{f}_j$ with a $\beta_{i,j}\times N$ matrix  $S_{i,j}$ of full rank,
\textit{i.e.}, $S_{i,j}\mathbf{f}_j$, where $S_{i,j}$  is  usually called the \textit{repair matrix} in the literature.
Similarly, we associate with each node $i$
 a rebuilding access profile
\begin{eqnarray*}
 \bm{\delta}_i\triangleq (\delta_{i,0},\delta_{i,1},\cdots,\delta_{i,i-1},\delta_{i,i+1},\cdots,\delta_{i,n-1}),
\end{eqnarray*}
where $\delta_{i,j}$  denotes the amount of symbols accessed at node $j$ when repairing node $i$, \textit{i.e.}, the number of nonzero columns of the matrix $S_{i,j}$.

\section{A Generic Transformation for MDS  Codes}\label{Sec tran}

In this section,  we propose a generic method that can transform any known nonbinary   $(n,k)$ MDS code  into a new $(n,k)$ MDS code with the optimal rebuilding access for an arbitrary set of $r=n-k$ nodes, while keeping the normalized repair bandwidth and the normalized
rebuilding access  of the other $k$ nodes intact. Given an $(n,k)$ base code,  the $r$ nodes which we wish to endow with the optimal repair property are called the \textit{target nodes}, while the other $k$ nodes are named the \textit{remainder nodes}. Without loss of generality, we always assume  that the last $r$ nodes are the target nodes unless otherwise stated. For simplicity,  sometimes we also denote by TN the target node and RN the remainder node in the sequel.
Before presenting this transformation,  an example is provided to illustrate the key idea behind it.

\subsection{An Example $(9,6)$ MDS  Code}
\label{ex_step2}
Given  a known nonbinary $(9,6)$ MDS code $\mathcal{C}_{1}$  over the   finite field $\mathbb{F}_q$,  where $q$ is odd (for the general construction, $q$ can be both even and odd),  let $S_{i,j}$, $j\in[0,9)\backslash\{i\}$ be the repair matrices for remainder node $i$ ($i\in[0,6)$).  For $l\in [0,3)$, let $\mathbf{f}_0^{(l)},\mathbf{f}_1^{(l)},\cdots,\mathbf{f}_{5}^{(l)}$ and $\mathbf{g}_0^{(l)},\mathbf{g}_1^{(l)},\mathbf{g}_2^{(l)}$ be the data  respectively stored at remainder nodes $0,1,\cdots,5$ and target nodes $0,1,2$ of an instance of the MDS code $\mathcal{C}_{1}$.   Through the generic transformation, we can obtain a $(9,6)$ MDS  code with the optimal rebuilding access for the target nodes, as given in Table \ref{Table stru new}.

\begin{table}[htbp]
\begin{center}
\caption{A $(9,6)$  MDS  code with the optimal rebuilding access for the target nodes}\label{Table stru new}
\setlength{\tabcolsep}{5.4pt}
\begin{tabular}{|c|c|c|c|c|c|c|}
\hline  RN 0  &  \multirow{2}{*}{$\cdots$} & RN $5$  & TN $0$ &TN $1$ & TN $2$  \\
($\mathbf{f}_0$) & &($\mathbf{f}_{5}$)&($\mathbf{f}_{6}$) & ($\mathbf{f}_{7}$)&($\mathbf{f}_{8}$)\\
\hline\hline $\mathbf{f}_0^{(0)}$ & $\cdots$ & $\mathbf{f}_{5}^{(0)}$ & $\mathbf{g}_0^{(0)}$ & $\uline{\mathbf{g}_{1}^{(0)}+\mathbf{g}_1^{(1)}}$ &  $\mathbf{g}_{2}^{(0)}+\mathbf{g}_2^{(2)}$\\
\hline  $\mathbf{f}_0^{(1)}$ & $\cdots$ & $\mathbf{f}_{5}^{(1)}$ &  \uline{${-\mathbf{g}_1^{(1)}+\mathbf{g}_{1}^{(0)}}$} & ${\mathbf{g}_{2}^{(1)}}$ &  $\mathbf{g}_{0}^{(1)}+\mathbf{g}_{0}^{(2)}$\\
\hline $\mathbf{f}_0^{(2)}$ & $\cdots$ & $\mathbf{f}_{5}^{(2)}$ & $\dashuline{-\mathbf{g}_2^{(2)}+\mathbf{g}_{2}^{(0)}}$ & $\dashuline{-\mathbf{g}_{0}^{(2)}+\mathbf{g}_{0}^{(1)}}$ &  $\mathbf{g}_{1}^{(2)}$\\
\hline
\end{tabular}
\end{center}
\end{table}

\textbf{Reconstruction:} Let us focus on the reconstruction of the original file by using data stored at nodes $2$ to $7$;  other cases can be addressed similarly.
In Table \ref{Table stru new}, from the symbols that are underlined, we can recover $\mathbf{g}_1^{(0)}$ and $\mathbf{g}_{1}^{(1)}$. Together with the other data in  rows 1, 2 and columns $2$ to $7$, we now have
\begin{align*}
&(\mathbf{f}_2^{(0)},\ldots,\mathbf{f}_5^{(0)}, \mathbf{g}_0^{(0)}, \mathbf{g}_1^{(0)}),\\
&(\mathbf{f}_2^{(1)},\ldots,\mathbf{f}_5^{(1)}, \mathbf{g}_1^{(1)}, \mathbf{g}_2^{(1)}),
\end{align*}
from which $(\mathbf{f}_0^{(0)},\ldots,\mathbf{f}_5^{(0)})$ and $(\mathbf{f}_0^{(1)},\ldots,\mathbf{f}_5^{(1)})$ can be reconstructed, respectively, because the base code is an MDS code.
Next, with these available data, $\mathbf{g}_2^{(0)}$ and $\mathbf{g}_0^{(1)}$ can now be computed, and then subtracted from the items marked with dashed underline to obtain $\mathbf{g}_2^{(2)}$ and $\mathbf{g}_0^{(2)}$.
Finally, together with the other data in the last row and columns $2$ to $5$,
we now also have
\begin{align*}
&(\mathbf{f}_2^{(2)},\ldots,\mathbf{f}_5^{(2)}, \mathbf{g}_0^{(2)}, \mathbf{g}_2^{(2)}),
\end{align*}
from which we can reconstruct $(\mathbf{f}_0^{(2)},\ldots,\mathbf{f}_5^{(2)})$. Thus the original file can indeed be reconstructed using data at nodes $2$ to $7$.

\textbf{Optimal rebuilding access for the target nodes:} Let us focus on the repair of target node $0$, for which the following data are downloaded
\begin{align*}
\mathbf{f}_{0}^{(0)},\mathbf{f}_{1}^{(0)},\ldots,\mathbf{f}_{5}^{(0)}, \mathbf{g}_{1}^{(0)}+\mathbf{g}_1^{(1)},\mathbf{g}_{2}^{(0)}+\mathbf{g}_2^{(2)},
\end{align*}
\textit{i.e.}, the data in row 1 of Table \ref{Table stru new}.
Clearly, $\mathbf{g}_0^{(0)}$ can be computed using $\mathbf{f}_{0}^{(0)},\mathbf{f}_{1}^{(0)},\ldots,\mathbf{f}_{5}^{(0)}$. To compute $-\mathbf{g}_1^{(1)}+\mathbf{g}_{1}^{(0)}$ stored at target node $0$, observe firstly that $\mathbf{g}_{1}^{(0)}$ can also be computed using  $\mathbf{f}_{0}^{(0)},\mathbf{f}_{1}^{(0)},\ldots,\mathbf{f}_{5}^{(0)}$, however, this implies that from the downloaded data $\mathbf{g}_{1}^{(0)}+\mathbf{g}_1^{(1)}$, we can recover $\mathbf{g}_1^{(1)}$ as well, and subsequently obtain $-\mathbf{g}_1^{(1)}+\mathbf{g}_{1}^{(0)}$. The other piece of coded data $-\mathbf{g}_2^{(2)}+\mathbf{g}_{2}^{(0)}$ stored at target node $0$ can be computed similarly. Thus  target node $0$ can indeed be repaired optimally and has the optimal rebuilding access.

\begin{table*}[htbp]
\begin{center}
\caption{Data downloaded from surviving nodes when repairing remainder node $0$ of the MDS code  in Table \ref{Table stru new}}\label{Table rep sys new}
\begin{tabular}{|c|c|c|c|c|c|}
\hline RN 1 ($\mathbf{f}_1$) & $\cdots$ & RN $5$ ($\mathbf{f}_{5}$) & TN $0$ ($\mathbf{f}_{6}$)  & TN $1$ ($\mathbf{f}_{7}$) & TN $2$ ($\mathbf{f}_{8}$) \\
\hline\hline  $S_{0,1}\mathbf{f}_1^{(0)}$ &  $\cdots$ & $S_{0,5}\mathbf{f}_{5}^{(0)}$ & {$S_{0,6}\mathbf{g}_0^{(0)}$} &  ${S_{0,7}(\mathbf{g}_{1}^{(0)}+\mathbf{g}_1^{(1)})}$ & $S_{0,8}(\mathbf{g}_{2}^{(0)}+\mathbf{g}_2^{(2)})$ \\
\hline  $S_{0,1}\mathbf{f}_1^{(1)}$ &   $\cdots$  &  $S_{0,5}\mathbf{f}_{5}^{(1)}$&  {$S_{0,7}(-\mathbf{g}_1^{(1)}+\mathbf{g}_{1}^{(0)})$} & ${S_{0,8}\mathbf{g}_{2}^{(1)}}$ &  $S_{0,6}(\mathbf{g}_{0}^{(1)}+\mathbf{g}_{0}^{(2)})$   \\
\hline
                 $S_{0,1}\mathbf{f}_1^{(2)}$ &   $\cdots$ &$S_{0,5}\mathbf{f}_{5}^{(2)}$&  $S_{0,8}(-\mathbf{g}_2^{(2)}+\mathbf{g}_{2}^{(0)})$ & $S_{0,6}(-\mathbf{g}_{0}^{(2)}+\mathbf{g}_{0}^{(1)})$& $S_{0,7}\mathbf{g}_{1}^{(2)}$\\
\hline
\end{tabular}
\end{center}
\end{table*}

\textbf{Repair efficiencies of the remainder nodes:}
Let us focus on repairing  remainder node $0$ of the constructed $(9,6)$ MDS code, which  can be accomplished by downloading the data in Table \ref{Table rep sys new}. To see this, consider the repair of $\mathbf{f}_{0}^{(0)}$, for which the original MDS code $\mathcal{C}_{1}$ needs to download
\begin{align}
S_{0,1}\mathbf{f}_1^{(0)}, S_{0,2}\mathbf{f}_2^{(0)},\ldots,S_{0,5}\mathbf{f}_5^{(0)}, S_{0,6}\mathbf{g}_0^{(0)},S_{0,7}\mathbf{g}_1^{(0)},S_{0,8}\mathbf{g}_2^{(0)}
\label{eqn:repairexample}
\end{align}
for the repair. Comparing these with the downloaded data in row 1 of Table \ref{Table rep sys new}, we know that $S_{0,7}\mathbf{g}_1^{(0)},S_{0,8}\mathbf{g}_2^{(0)}$ are not directly available. However, ${S_{0,7}(\mathbf{g}_{1}^{(0)}+\mathbf{g}_1^{(1)})}$, downloaded from target node $1$, and $S_{0,7}(-\mathbf{g}_1^{(1)}+\mathbf{g}_{1}^{(0)})$, downloaded from target node $0$, can be utilized to recover $S_{0,7}\mathbf{g}_1^{(0)}$; the data $S_{0,8}\mathbf{g}_2^{(0)}$ can be recovered similarly. At this point, with all the data listed in (\ref{eqn:repairexample}) available, the repair mechanism in the original MDS code $\mathcal{C}_{1}$ can be invoked to compute $\mathbf{f}_{0}^{(0)}$. The repair of $\mathbf{f}_{0}^{(1)}$ and $\mathbf{f}_{0}^{(2)}$ can be done in a similar manner, and thus  remainder  node $0$ can indeed be repaired.

Now, let us investigate the repair efficiencies of remainder node 0, \textit{i.e.}, the normalized repair bandwidth and the normalized rebuilding access. Let $\bm{\beta}_{0}$ and  $\bm{\delta}_{0}$  (resp. $\bm{\hat{\beta}}_{0}$ and $\bm{\hat{\delta}}_{0}$) respectively be the repair bandwidth profile and the rebuilding access profile of remainder node 0 of the base code (resp. the new code). From the above analysis, it is easy to see that
\begin{eqnarray}\label{Eqn ex profiles}
  \sum\limits_{j=1}^{8}\hat{\beta}_{0,j} = 3 \sum\limits_{j=1}^{8}\beta_{0,j},~
\sum\limits_{j=1}^{8}\hat{\delta}_{0,j} = 3 \sum\limits_{j=1}^{8}\delta_{0,j}.
\end{eqnarray}
Note that the file size of the new code is three times as that of   the base code, which in conjunction with \eqref{Eqn ex profiles} implies that
remainder node 0 of the new code has the same  normalized repair bandwidth and normalized rebuilding access as those of  the base code.
The repair efficiencies of the other remainder nodes can be verified in the same manner.

\subsection{The Generic Transformation}\label{subsec generic trans}
In this subsection, we present the generic transformation, which utilizes a known  nonbinary $(n,k)$ MDS code $\mathcal{C}_{1}$  with a sub-packetization level $N$ as the base code.   Let     $\bm{\beta}_i$ and $\bm{\delta}_i$ respectively denote the repair bandwidth profile  and rebuilding access profile for node $i$.
The transformation can be performed through the following three steps.

\subsection*{\textbf{Step 1: An intermediate MDS code $\mathcal{C}_2$ by SPACE SHARING $r$ instances of the base code $\mathcal{C}_{1}$ }}
Let
$\mathbf{f}_{i}^{(l)}$ and $\mathbf{g}_{j}^{(l)}$ respectively be the data stored at remainder node $i$ and target node $j$ of an instance of the code $\mathcal{C}_{1}$,   where $i\in [0,k)$ and $l,j\in [0,r)$. We can thus construct an intermediate MDS  code $\mathcal{C}_2$ with sub-packetization level $rN$ by space sharing $r$ instances of the base code $\mathcal{C}_{1}$.

\subsection*{\textbf{Step 2: An intermediate MDS code  $\mathcal{C}_3$  by PERMUTING the data in the target nodes of $\mathcal{C}_2$}}

From $\mathcal{C}_2$, we construct another intermediate MDS  code $\mathcal{C}_3$ by permuting the data in the target nodes while keeping the remainder
nodes intact. Let  $\mathbf{h}_{j}$ denote the data stored at target node $j$ of code $\mathcal{C}_3$. For convenience, we write $\mathbf{h}_{j}$ as \begin{eqnarray*}
  \mathbf{h}_{j} &=& \left(
                         \begin{array}{c}
                           \mathbf{h}_{j}^{(0)} \\
                           \vdots \\
                           \mathbf{h}_{j}^{(r-1)} \\
                         \end{array}
                       \right), j\in [0,r)
\end{eqnarray*}
where $\mathbf{h}_{j}^{(l)}$ ($l\in [0,r)$) is a column vector of length $N$. Let $\pi_0,\pi_1,\cdots,\pi_{r-1}$ be $r$ permutations on $[0,r)$, which should satisfy some specific requirements (the requirements are given more precisely in Theorem \ref{Thm repair sys}). Then $\mathbf{h}_{j}^{(l)}$ in $\mathcal{C}_3$ is defined as
\begin{eqnarray}\label{Eqn perm h}
\mathbf{h}_{j}^{(l)}=\mathbf{g}_{\pi_l(j)}^{(l)}, ~~j,l\in [0,r).
\end{eqnarray}

\subsection*{\textbf{Step 3: The resultant  storage code $\mathcal{C}_4$  by PAIRING the data in the target nodes of $\mathcal{C}_3$}}

From the  code  $\mathcal{C}_3$, we construct the desired storage code $\mathcal{C}_4$ by modifying only the data at the target nodes while keeping the remainder nodes   intact. Let   $\mathbf{h}'_{j}$ denote the data stored at target node $j$  of code $\mathcal{C}_4$. For convenience, we write $\mathbf{h}'_{j}$ as
\begin{eqnarray*}
  \mathbf{h}'_{j} &=& \left(
                         \begin{array}{c}
                           \mathbf{h'}_{j}^{(0)} \\
                           \vdots \\
                           \mathbf{h'}_{j}^{(r-1)} \\
                         \end{array}
                       \right), j\in [0,r)
\end{eqnarray*}
where $\mathbf{h'}_{j}^{(l)}$ ($l\in [0,r)$) is a column vector of length $N$
defined by
\begin{eqnarray}\label{Eqn new con}
\mathbf{h'}_{j}^{(l)}=\left\{ \begin{array}{ll}
\mathbf{h}_{j}^{(j)}, & \textrm{if}~j=l \\
\theta_{j,l}\mathbf{h}_{j}^{(l)}+\eta_{l,j}\mathbf{h}_{l}^{(j)}, &  \mbox{otherwise}
\end{array}\right.
\end{eqnarray}
with $\theta_{j,l},~\eta_{l,j}\in \mathbf{F}_q\backslash\{0\}$ such that $\mathbf{h'}_{j}^{(l)}$ and $\mathbf{h'}_{l}^{(j)}$ are linearly independent for $j\ne l$. Particularly,  we can set $\eta_{l,j}=1$
\begin{eqnarray}\label{Eqn def theta}
  \{\theta_{j,l},\theta_{l,j}\}=\{1,a\}
\end{eqnarray}
for all $j,l\in [0,r)~\mbox{with}~j\ne l$ and $a\in \mathbb{F}_q\backslash\{0,1\}$ for convenience, which can also guarantee the pairwise equations
\begin{eqnarray}\label{Eqn Pairwise}
\left\{ \begin{array}{l}
\theta_{j,l}\mathbf{h}_{j}^{(l)}+\mathbf{h}_{l}^{(j)}=\mathbf{h'}_{j}^{(l)}\\
\theta_{l,j}\mathbf{h}_{l}^{(j)}+\mathbf{h}_{j}^{(l)}=\mathbf{h'}_{l}^{(j)}
\end{array}\right.
\end{eqnarray}
are linearly independent.

The  new code $\mathcal{C}_4$ is  depicted in Table  \ref{Table C4}.

\begin{table*}[htbp]
\begin{center}
\caption{The new storage code  $\mathcal{C}_4$}\label{Table C4}
\begin{tabular}{|c|c|c|c|c|c|c|c|}
\hline   RN 0 & $\cdots$  & RN $k-1$ & TN $0$ ($\mathbf{h'}_{0}$)& TN $1$ ($\mathbf{h}'_{1}$)& $\cdots$  &TN $r-1$ ($\mathbf{h'}_{r-1}$)\\
\hline\hline $\mathbf{f}_0^{(0)}$ &  $\cdots$ & $\mathbf{f}_{k-1}^{(0)}$ &$\mathbf{h}_0^{(0)}$ &  $\theta_{1,0}\mathbf{h}_1^{(0)}+\mathbf{h}_0^{(1)}$ &
$\cdots$ & $\theta_{r-1,0}\mathbf{h}_{r-1}^{(0)}+\mathbf{h}_0^{(r-1)}$ \\
\hline  $\mathbf{f}_0^{(1)}$ &  $\cdots$  &  $\mathbf{f}_{k-1}^{(1)}$& $\theta_{0,1}\mathbf{h}_0^{(1)}+\mathbf{h}_1^{(0)}$&$\mathbf{h}_1^{(1)}$&  $\cdots$  & $\theta_{r-1,1}\mathbf{h}_{r-1}^{(1)}+\mathbf{h}_1^{(r-1)}$ \\
\hline  $\vdots$&$\ddots$ & $\vdots$&$\vdots$&$\vdots$&$\ddots$ &$\vdots$\\
\hline
  $\mathbf{f}_0^{(r-1)}$ & $\cdots$ & $\mathbf{f}_{k-1}^{(r-1)}$ & $\theta_{0,r-1}\mathbf{h}_0^{(r-1)}+\mathbf{h}_{r-1}^{(0)}$& $\theta_{1,r-1}\mathbf{h}_1^{(r-1)}+\mathbf{h}_{r-1}^{(1)}$ & $\cdots$ & $\mathbf{h}_{r-1}^{(r-1)}$ \\
\hline
\end{tabular}
\end{center}
\end{table*}

We next show that the MDS property holds for the new $(n,k)$ storage code $\mathcal{C}_4$.

\begin{Theorem}\label{Thm MDS C4}
Code $\mathcal{C}_4$ has the MDS property.
\end{Theorem}

\begin{table*}[htbp]
\begin{center}
\caption{
}\label{Table Thm MDS 1}
\begin{tabular}{|c|c|c|c|c|c|c|c|c|}
\hline $\mathbf{h}_{j_0}^{(j_0)}$& $\theta_{j_1,j_0}\mathbf{h}_{j_1}^{(j_0)}+\mathbf{h}_{j_0}^{(j_1)}$ &$\cdots$ & $\theta_{j_{t-1},j_0}\mathbf{h}_{j_{t-1}}^{(j_0)}+\mathbf{h}_{j_0}^{(j_{t-1})}$\\
\hline $\theta_{j_0,j_1}\mathbf{h}_{j_0}^{(j_1)}+\mathbf{h}_{j_1}^{(j_0)}$& $\mathbf{h}_{j_1}^{(j_1)}$& $\cdots$ & $\theta_{j_{t-1},j_1}\mathbf{h}_{j_{t-1}}^{(j_1)}+\mathbf{h}_{j_1}^{(j_{t-1})}$ \\
\hline  $\vdots$ & $\vdots$ &$\ddots$ &$\vdots$ \\
\hline
$\theta_{j_0,j_{t-1}}\mathbf{h}_{j_0}^{(j_{t-1})}+\mathbf{h}_{j_{t-1}}^{(j_0)}$  &  $\theta_{j_1,j_{t-1}}\mathbf{h}_{j_1}^{(j_{t-1})}+\mathbf{h}_{j_{t-1}}^{(j_1)}$
    & $\vdots$  &
 $\mathbf{h}_{j_{t-1}}^{(j_{t-1})}$ \\
\hline
\end{tabular}
\end{center}
\end{table*}

\begin{table*}[htbp]
\begin{center}
\caption{
}\label{Table Thm MDS 3}
\begin{tabular}{|c|c|c|c|}
\hline $\theta_{j_0,j_t}\mathbf{h}_{j_0}^{(j_t)}+\dashuline{\mathbf{h}_{j_t}^{(j_0)}}$& $\theta_{j_1,j_t}\mathbf{h}_{j_1}^{(j_t)}+\dashuline{\mathbf{h}_{j_t}^{(j_1)}}$ & $\cdots$ & $\theta_{j_{t-1},j_t}\mathbf{h}_{j_{t-1}}^{(j_t)}+\dashuline{\mathbf{h}_{j_t}^{(j_{t-1})}}$\\
\hline $\vdots$  & $\vdots$ & $\ddots$  & $\vdots$ \\
\hline
 $\theta_{j_0,j_{r-1}}\mathbf{h}_{j_0}^{(j_{r-1})}+\dashuline{\mathbf{h}_{j_{r-1}}^{(j_0)}}$  & $\theta_{j_1,j_{r-1}}\mathbf{h}_{j_1}^{(j_{r-1})}+\dashuline{\mathbf{h}_{j_{r-1}}^{(j_1)}}$  &
  $\vdots$ & $\theta_{j_{t-1},j_{r-1}}\mathbf{h}_{j_{t-1}}^{(j_{r-1})}+\dashuline{\mathbf{h}_{j_{r-1}}^{(j_{t-1})}}$ \\
\hline
\end{tabular}
\end{center}
\end{table*}

\begin{proof}
The  code $\mathcal{C}_4$ possesses the MDS property if any $k$ out of the $n$ nodes can reconstruct the original file, which is equivalent to reconstructing the data $\mathbf{f}_{i}^{(l)}$, $i\in [0,k)$ and $l\in [0,r)$  at the remainder nodes according to the MDS property of the base code.
We discuss the reconstruction in two cases.
\begin{itemize}

\item [(i)] When connecting  to all the $k$ remainder nodes: there is nothing to prove.

\item [(ii)] When connecting to $k-t$ remainder nodes and $t$ target nodes where $1\le t\le \min\{r,k\}$:  we assume that $I=\{i_0,i_1,\cdots,i_{t-1}\}$  is  the set of the indices of  the remainder nodes which are not connected and $J=\{j_0,j_1,\cdots,j_{t-1}\}$ is the set of the indices of the target nodes which are  connected, where $0\le i_0<\cdots<i_{t-1}<k$ and $0\le j_0<\cdots<j_{t-1}<r$. Denote  $\{j_t,\cdots,j_{r-1}\}=[0,r)\backslash J$.

Firstly, given the data in Table \ref{Table Thm MDS 1}   from the  target nodes connected,  we can obtain the data $\mathbf{h}_{u}^{(l)}$ ($l,u\in J$)   by solving pairwise  linearly independent equations as \eqref{Eqn Pairwise} (specifically for $t=1$, no equation needs to be solved). Secondly, for each $l\in J$, combining the  data $\mathbf{h}_{u}^{(l)}$ ($u\in J$) at the target nodes with the  data $\mathbf{f}_{i}^{(l)}$ ($i\in [0,k-1]\backslash I$) at the $k-t$ remainder nodes of code $\mathcal{C}_4$ connected, we can obtain $\mathbf{h}_{u}^{(l)}$, $u\in [0,r)\backslash J$, by means of the MDS property of the base code $\mathcal{C}_{1}$ and \eqref{Eqn perm h}. Thirdly, from the data in Table \ref{Table Thm MDS 3} at the  target nodes connected, we then are able to obtain the data $\mathbf{h}_{u}^{(l)}$ ($l\in [0,r)\backslash J$, $u\in J$) by eliminating the terms  $\mathbf{h}_{l}^{(u)}$ ($u\in J$, $l\in [0,r)\backslash J$) marked with dash underline. That is, for each $l\in [0,r)$ and $u\in J$, the data $\mathbf{h}_{u}^{(l)}$, \textit{i.e.}, $\mathbf{g}_{\pi_l(u)}^{(l)}$, is available. Finally, together with $\mathbf{f}_{i}^{(l)}$, $i\in [0,k-1]\backslash I$  at the $k-t$  remainder nodes connected,  we can recover the remaining data $\mathbf{f}_{i_0}^{(l)},\cdots,\mathbf{f}_{i_{t-1}}^{(l)}$ by means of the MDS property of the base code $\mathcal{C}_{1}$ for each $l\in [0,r)$.
\end{itemize}
\end{proof}

Next, we verify that the target nodes of  code $\mathcal{C}_4$ have the optimal repair bandwidth and the optimal rebuilding access.

\begin{Theorem}\label{Thm repair parity}
Target node $j$ ($j\in[0,r)$) in code $\mathcal{C}_4$ has the optimal repair bandwidth and the optimal rebuilding access. Specifically,  the
  repair bandwidth profile $\bm{\hat{\beta}}_{k+j}$ and the rebuilding access profile $\bm{\hat{\delta}}_{k+j}$  are given by
\begin{eqnarray*}
  \hat{\beta}_{k+j,i}= \hat{\delta}_{k+j,i}=N,~~ i\in[0,n)\backslash\{k+j\}.
\end{eqnarray*}
\end{Theorem}

\begin{proof}
We show that for any $j\in [0,r)$,  target node $j$ can be repaired by accessing and
downloading $\mathbf{h'}^{(j)}_{l}$, $l\in [0,r)\backslash\{j\}$, and $\mathbf{f}^{(j)}_{i}$, $i\in [0,k)$.

Firstly, using $\mathbf{f}^{(j)}_{i}$, $i\in [0,k)$, we can compute $\mathbf{g}^{(j)}_{s}$, $s\in [0,r)$, and then obtain $\mathbf{h}^{(j)}_{s}$, $s\in [0,r)$, according to \eqref{Eqn perm h}.
Next, for any $l\in [0,r)\backslash\{j\}$, from the downloaded data $\mathbf{h'}^{(j)}_{l}=\theta_{l,j} \mathbf{h}^{(j)}_{l}+ \mathbf{h}^{(l)}_{j}$, we can obtain $\mathbf{h}^{(l)}_{j}$ by subtracting $\mathbf{h}^{(j)}_{l}$ from $\mathbf{h'}^{(j)}_{l}$,  and thus  $\mathbf{h'}^{(l)}_{j}=\theta_{j,l} \mathbf{h}^{(l)}_{j}+ \mathbf{h}^{(j)}_{l}$. Finally, since $\mathbf{h'}^{(j)}_{j}=\mathbf{h}^{(j)}_{j}$, which has already been computed in the first step, target node $j$ can indeed be repaired optimally.

Applying the definitions of the repair bandwidth (profile) and the rebuilding access (profile), we obtain the desired result.
\end{proof}

Finally,  we examine the  repair of the remainder nodes of code $\mathcal{C}_4$, which will be proceeded in two cases, according to whether the repair strategy for a  remainder node  of the base code is \textit{naive} or not. Naive  repair means that a node is repaired by download all the data from any $k$ surviving nodes to first reconstruct the original file, and then repair the failed node.  Particularly,  the repair strategy of the remainder nodes of code $\mathcal{C}_4$ is almost the same as that of the base code.

\begin{Theorem}\label{Thm repair sys}
For each $i\in [0,k)$, remainder node $i$ of the $(n,k)$ MDS code $\mathcal{C}_4$  has  the same  normalized repair
bandwidth and rebuilding access as those of  the base code if
\begin{itemize}
\item [(i)] The repair strategy for remainder node $i$ of the base code is naive, or
\item [(ii)] There exists some matrix $S_{i}$ such that $S_{i,k+j}=S_{i}$ for all $j\in [0,r)$, or
\item [(iii)] $\pi_l(j)=\pi_j(l)$ for $l,j\in [0,r)$.
\end{itemize}
\end{Theorem}

\begin{proof}
If the repair strategy for  remainder node $i$ of the base code is naive, then remainder node $i$ of code $\mathcal{C}_4$ can also be naively repaired  due to the  MDS property of
code $\mathcal{C}_4$.

Let us now focus on the general case. Recall  from the repair mechanism of the base code that, for $l\in [0,r)$, $\mathbf{f}_i^{(l)}$ can be obtained by the data $S_{i,s}\mathbf{f}_s^{(l)}$, $s\in [0,k)\backslash\{i\}$, and $S_{i,k+j}\mathbf{g}_{j}^{(l)}$, $j\in [0,r)$.
If there exists a matrix $S_i$ such that $S_i=S_{i,k+j}$ for all $j\in [0,r)$, or $\pi_l(j)=\pi_j(l)$ for $l,j\in [0,r)$, then
\begin{eqnarray}\label{Eqn constant S}
S_{i,k+\pi_l(j)}=S_{i,k+\pi_j(l)}
\end{eqnarray}
for all $j,l\in [0,r)$ with $j\ne l$.
The repair process for remainder node $i$ of code $\mathcal{C}_4$ can be repaired using the following three steps:
\begin{itemize}
\item [(a)] Download   $S_{i,s}\mathbf{f}_{s}^{(l)}$ and $S_{i,k+\pi_l(j)}\mathbf{h'}_{j}^{(l)}$ with $s\in [0,k)\backslash\{i\}$ and $j,l\in [0,r)$,
 \item [(b)] For all $j,l\in [0,r)$ with $j\ne l$, according to   \eqref{Eqn constant S}, \eqref{Eqn perm h} and \eqref{Eqn new con}, we can compute
      \begin{eqnarray*}\label{Eqn_Sijk}
 &&\theta_{l,j}S_{i,k+\pi_l(j)}\mathbf{h'}_{j}^{(l)}-S_{i,k+\pi_j(l)}\mathbf{h'}_{l}^{(j)}\nonumber\\
 &=&(\theta_{l,j}\theta_{j,l}-1)S_{i,k+\pi_l(j)}\mathbf{h}_{j}^{(l)}\\
 &=&(\theta_{l,j}\theta_{j,l}-1)S_{i,k+\pi_l(j)}\mathbf{g}_{\pi_l(j)}^{(l)}
 \end{eqnarray*}
  to obtain $S_{i,k+\pi_l(j)}\mathbf{g}_{\pi_l(j)}^{(l)}$.
  \item [(c)] For each $l\in [0,r)$, invoke the repair procedure of the base MDS code to regenerate $\mathbf{f}_i^{(l)}$ by the data $S_{i,s}\mathbf{f}_s^{(l)}$, $s\in [0,k)\backslash\{i\}$, and $S_{i,k+j}\mathbf{g}_{j}^{(l)}$, $j\in [0,r)$.
\end{itemize}
The above analysis, together with the fact that the  sub-packetization level of code $\mathcal{C}_4$ is $r$ times as that of the base code, implies the desired result.
\end{proof}

\begin{Corollary}
If the repair strategy for remainder node $i$ of the base code is naive, then the repair
bandwidth profile $\bm{\hat{\beta}}$ and the rebuilding access profile $\bm{\hat{\delta}}$ of remainder node $i$ of  the code $\mathcal{C}_4$ statisfy
\begin{eqnarray*}
\sum\limits_{j=0,j\ne i}^{n-1}\hat{\beta}_{i,j}=
r\sum\limits_{j=0,j\ne i}^{n-1}\beta_{i,j},
~~\sum\limits_{j=0,j\ne i}^{n-1}\hat{\delta}_{i,j}=
r\sum\limits_{j=0,j\ne i}^{n-1}\delta_{i,j};
\end{eqnarray*}
Otherwise, we have
\begin{eqnarray*}
\hat{\beta}_{i,j}=\left\{ \begin{array}{ll}
r\beta_{i,j},& \mbox{if}~j\in[0,k)\backslash\{i\}\\
\sum\limits_{l=0}^{r-1}\beta_{i,k+l},&\mbox{otherwise}
\end{array}\right.
\end{eqnarray*}
and
\begin{eqnarray*}
\hat{\delta}_{i,j}=\left\{ \begin{array}{ll}
r\delta_{i,j},& \mbox{if}~j\in[0,k)\backslash\{i\}\\
\sum\limits_{l=0}^{r-1}\delta_{i,k+l},&\mbox{otherwise}
\end{array}\right..
\end{eqnarray*}

Consequently, if a remainder node has the optimal repair bandwidth   or the optimal
rebuilding access  in the base code, the resultant code
$\mathcal{C}_4$ will maintain the same optimality.
\end{Corollary}

\begin{Remark}
Note that in all the aforementioned  $(n,k)$  MDS codes \cite{invariant-subspace,hadamard,Sasidharan-Kumar,Long-IT,extend-zigzag} except the Zigzag code \cite{Zigzag}, simple repair matrices with the form $S_{i,j}=S_i$ are used. In fact, it was shown in \cite{Access} that any systematic    MDS code that can optimally repair the systematic nodes  can be transformed into another   MDS code with such simple repair matrices, however at a cost of sacrificing a systematic node. The proposed generic transformation is valid for general repair matrices $S_{i,j}$, but the repair strategies for the remainder nodes exhibit different flexibilities if the condition $S_{i,k+j}=S_i$ holds for   all $j\in [0,r)$, i.e., in this case the permutations can be arbitrary as shown in Theorem \ref{Thm repair sys} item (ii).
\end{Remark}

\begin{Remark}
There are many choices of the permutations $\pi_0,\cdots,\pi_{r-1}$ satisfying the condition $\pi_i(j)=\pi_j(i)$ in Theorem \ref{Thm repair sys}, for example,
\begin{eqnarray*}
  \pi_i(j) = (i+j)~mod~r,~~ i,j\in [0,r),
\end{eqnarray*}
as we used in Section  \ref{ex_step2}.
\end{Remark}

\subsection{A Substitution Technique for Step 3 -  Target Nodes Unchanged}

In step 3 of the generic transformation in Section \ref{subsec generic trans}, we modified the data at the $r$ target nodes of code $\mathcal{C}_3$ to endow them with the optimal repair property.
However,  the resultant code $\mathcal{C}_4$ is no longer of systematic form if some $r$ systematic nodes are chosen as the target nodes. In this subsection, we provide an alternative solution, which endows any $r$ target nodes with the optimal repair property, but maintaining the systematic form of the code. This alternative approach allows us to modify the data at some $r$ remainder  nodes by pairing the target nodes' data components at these nodes, essentially substituting the original pairing operation on the target nodes.

Without loss of generality,  we choose  the last nodes as target nodes and modify the data at the first $r$ nodes. Recall that the base code $\mathcal{C}_{1}$ is an MDS code, which implies that $\mathbf{f}_{0}^{(l)}, \cdots, \mathbf{f}_{r-1}^{(l)}$ can be represented by $\mathbf{f}_{r}^{(l)}, \cdots, \mathbf{f}_{k-1}^{(l)}, \mathbf{g}_{0}^{(l)}, \cdots$, $\mathbf{g}_{r-1}^{(l)}$ for any $l\in [0,r)$. That is,
\begin{eqnarray*}
  \mathbf{f}_{j}^{(l)} &=& \sum\limits_{t=r}^{k-1}A_{j,t}\mathbf{f}_{t}^{(l)}+\sum\limits_{t=0}^{r-1} A_{j,t}\mathbf{g}_{t}^{(l)}\\
  &=&\sum\limits_{t=r}^{k-1}A_{j,t}\mathbf{f}_{t}^{(l)}+\sum\limits_{t=0}^{r-1}A_{j,\pi_l(t)}\mathbf{h}_{t}^{(l)}, ~j,~l\in [0,r),
\end{eqnarray*}
for some nonsingular matrices $A_{j,0},\cdots,A_{j,k-1}$ of order $N$, where the second equality follows from \eqref{Eqn perm h}. Based on the MDS code  $\mathcal{C}_3$ and \eqref{Eqn new con}, we can define a new  storage code $\mathcal{C}_4'$ as given in Table \ref{Table C4'},
where
\begin{eqnarray}\label{Eqn f'}
  \mathbf{f'}_{j}^{(l)} = \sum\limits_{t=r}^{k-1}A_{j,t}\mathbf{f}_{t}^{(l)}+\sum\limits_{t=0}^{r-1}A_{j,\pi_l(t)}\mathbf{h'}_{t}^{(l)},  ~j,~l\in [0,r).
\end{eqnarray}

\begin{table*}[htbp]
\begin{center}
\caption{New storage code  $\mathcal{C}_4'$}\label{Table C4'}
\begin{tabular}{|c|c|c|c|c|c|c|c|c|}
\hline  RN 0  & $\cdots$  & RN   $r-1$   &RN   $r$ & $\cdots$  & RN   $k-1$ & TN   $k$ &$\cdots$  & TN   $n-1$ \\
\hline\hline   $\mathbf{f'}_0^{(0)}$ & $\cdots$ &$\mathbf{f'}_{r-1}^{(0)}$ &$\mathbf{f}_{r}^{(0)}$ & $\cdots$ & $\mathbf{f}_{k-1}^{(0)}$ & $\mathbf{h}_{0}^{(0)}$ & $\cdots$ & $\mathbf{h}_{r-1}^{(0)}$ \\
\hline   $\mathbf{f'}_0^{(1)}$ & $\cdots$  & $\mathbf{f'}_{r-1}^{(0)}$ &$\mathbf{f}_{r}^{(1)}$&$\cdots$  &$\mathbf{f}_{k-1}^{(1)}$&  $\mathbf{h}_{0}^{(1)}$& $\cdots$  &$\mathbf{h}_{r-1}^{(1)}$ \\
\hline   $\vdots$ & $\ddots$ & $\vdots$ & $\vdots$ &$\ddots$ &$\vdots$ &$\vdots$ & $\ddots$ & $\vdots$ \\
\hline
$\mathbf{f'}_0^{(r-1)}$ &   $\cdots$&   $\mathbf{f'}_{r-1}^{(0)}$ & $\mathbf{f}_{r}^{(r-1)}$  &  $\cdots$& $\mathbf{f}_{k-1}^{(r-1)}$ &  $\mathbf{h}_{0}^{(r-1)}$&  $\cdots$& $\mathbf{h}_{r-1}^{(r-1)}$ \\
\hline
\end{tabular}
\end{center}
\end{table*}

\begin{table*}[htbp]
\begin{center}
\caption{The storage  code  $\mathcal{C}_3'$}\label{Table C3 II}
\begin{tabular}{|c|c|c|c|c|c|c|c|c|}
\hline RN 0  & $\cdots$  & RN   $r-1$   &RN   $r$ & $\cdots$  & RN   $k-1$ & TN   $k$ &$\cdots$  & TN   $n-1$ \\
\hline\hline $\mathbf{f'}_0^{(0)}$ & $\cdots$ &$\mathbf{f'}_{r-1}^{(0)}$ &$\mathbf{f}_{r}^{(0)}$ & $\cdots$ & $\mathbf{f}_{k-1}^{(0)}$ & $\mathbf{h'}_{0}^{(0)}$ & $\cdots$ & $\mathbf{h'}_{r-1}^{(0)}$ \\
\hline   $\mathbf{f'}_0^{(1)}$ & $\cdots$  & $\mathbf{f'}_{r-1}^{(0)}$ &$\mathbf{f}_{r}^{(1)}$&$\cdots$  &$\mathbf{f}_{k-1}^{(1)}$&  $\mathbf{h'}_{0}^{(1)}$& $\cdots$  &$\mathbf{h'}_{r-1}^{(1)}$ \\
\hline   $\vdots$ & $\ddots$ & $\vdots$ & $\vdots$ &$\ddots$ &$\vdots$ &$\vdots$ & $\ddots$ & $\vdots$ \\
\hline
  $\mathbf{f'}_0^{(r-1)}$ &
   $\cdots$&
   $\mathbf{f'}_{r-1}^{(0)}$ &
 $\mathbf{f}_{r}^{(r-1)}$  &
  $\cdots$&
 $\mathbf{f}_{k-1}^{(r-1)}$ &
  $\mathbf{h'}_{0}^{(r-1)}$&
  $\cdots$&
 $\mathbf{h'}_{r-1}^{(r-1)}$ \\
\hline
\end{tabular}
\end{center}
\end{table*}

Note from \eqref{Eqn new con} that
\begin{eqnarray}\label{Eqn new h'}
\mathbf{h}_{j}^{(l)}=\left\{ \begin{array}{ll}
\mathbf{h'}_{j}^{(j)}, & \textrm{if}~j=l \\
\theta'_{j,l}\mathbf{h'}_{j}^{(l)}+\eta'_{l,j}\mathbf{h'}_{l}^{(j)}, &  \mbox{otherwise}
\end{array}\right.
\end{eqnarray}
where
\begin{eqnarray*}
  \theta'_{j,l} = {\theta_{l,j}\over \theta_{l,j}\theta_{j,l}-1},~~\eta'_{l,j} = {-1\over \theta_{l,j}\theta_{j,l}-1}.
\end{eqnarray*}
In this sense,  the new code $\mathcal{C}_4'$ can be obtained by pairing the data at the target nodes of the storage code $\mathcal{C}_3'$  in Table \ref{Table C3 II}, \textit{i.e.}, by applying step 3 to the code $\mathcal{C}_3'$.

It is obvious that code $\mathcal{C}_3'$ has the MDS property. Then, following the proofs of  Theorems  \ref{Thm MDS C4}-\ref{Thm repair parity}
 we  immediately have a  corollary.

\begin{Corollary}
 Code $\mathcal{C}_4'$ has the MDS property and the same repair property as that of code $\mathcal{C}_4$.
\end{Corollary}

\begin{Remark}\label{Rm_compare}
The formula  \eqref{Eqn new con} is also the key  technique used in \cite{Sasidharan-Kumar2} and \cite{Barg2},
which is named pairwise coupling transformation (PCT) in \cite{tight-bound-on-alpha}. As seen from the three steps of our generic transformation, in addition to the main conceptual differences discussed in Section \ref{sec:history},
a few more subtle differences are that 1) our generic transformation is valid for both scalar MDS codes and vector MDS codes, while the PCT in \cite{Sasidharan-Kumar2,Barg2} only aims for scalar MDS codes; and 2)
the proposed transformation is described in three simple steps and is more flexible, particularly,
\begin{itemize}
  \item [(i)] The permutations in step $2$ can be arbitrary in some cases;
\item [(ii)] The data modification in step $3$ can be performed on any $r$ target nodes, or any other $r$ remainder nodes. As a consequence, the resultant  MDS code can keep its systematic form.
\end{itemize}
\end{Remark}

\section{Applications of the generic transformation}

In the previous section, we provided a generic method that can transform any known nonbinary   $(n,k)$ MDS codes into a new $(n,k)$ MDS code with the optimal rebuilding access for an arbitrary set of $r$ nodes while preserving the  normalized repair bandwidth and the  normalized
rebuilding access  of the other $k$ nodes.
In this section,
we discuss two specific applications of the transformation, which
provide solutions to two long standing problems in this area.

\subsection{Constructing All-Node-Repair MDS Codes}
Clearly, if we start with a base nonbinary MDS code
$\mathcal{C}_1$ which has the optimal repair bandwidth (or the optimal rebuilding
access) for the systematic nodes only, such as the MDS codes constructed in  \cite{hadamard,Hadamard-strategy,Zigzag,invariant-subspace,Long-IT}, we can apply the
transformation by taking the parity nodes as the target
nodes, and obtain an MDS codes $\mathcal{C}_4$ with the optimal repair
bandwidth (or the optimal rebuilding access) for both the systematic
nodes and the parity nodes. Moreover, $\mathcal{C}_4$ uses the same finite field
as the base code, and has a sub-packetization level a factor of $r$ as
large as that of the base code.  

\subsection{Building Optimal Repair Codes from Scalar MDS Codes}

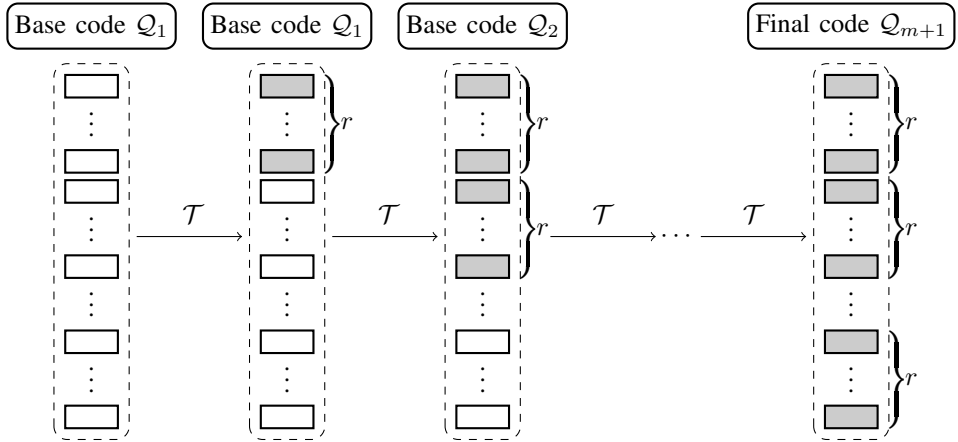
\begin{figure*}[!htbp]
\tikzset{
source/.style={draw,thick,inner sep=.1cm,minimum width=2.5cm},
source1/.style={draw,thick,inner sep=.1cm,minimum width=0.7cm,minimum height=.3cm,fill=gray!40},
source3/.style={draw,thick,inner sep=.1cm,minimum width=0.7cm,minimum height=.3cm},
source2/.style={draw,thick,rounded corners,inner sep=.1cm,minimum height=.6cm},
}
\centering
  \begin{tikzpicture}[every node/.style={inner sep=4pt}, scale=1, every node/.style={scale=1}]
\node [source3] (node 0) at (0,0)  {};
		  \node at (0,-0.4)  {$\vdots$};
		\node [source3] (node r-1) at (0,-1)  {};
		\node [source3] (node 0) at (0,-1.4)  {};
		  \node at (0,-1.8)  {$\vdots$};
		\node [source3] (node r-1) at (0,-2.4)  {};
\node at (0,-2.8)  {{$\vdots$}};
\node [source3] (node 0) at (0,-3.4)  {};
		  \node at (0,-3.8)  {$\vdots$};
		\node [source3] (node r-1) at (0,-4.4)  {};
		
\draw[dashed,rounded corners=5pt] (-.5,-4.7) rectangle (.5,0.3);
\node [source1] (node 0) at (2.6,0)  {};
		  \node at (2.6,-0.4)  {$\vdots$};
		\node [source1] (node r-1) at (2.6,-1)  {};
		\node [source3] (node 0) at (2.6,-1.4)  {};
		  \node at (2.6,-1.8)  {$\vdots$};
		\node [source3] (node r-1) at (2.6,-2.4)  {};
\node at (2.6,-2.8)  {{$\vdots$}};
\node [source3] (node 0) at (2.6,-3.4)  {};
		  \node at (2.6,-3.8)  {$\vdots$};
		\node [source3] (node r-1) at (2.6,-4.4)  {};
\draw[dashed,rounded corners=5pt] (2.1,-4.7) rectangle (3.1,0.3);

\draw[->] (.6,-2) -- (2,-2);	
\node [] (Transformation) at (1.35,-1.7)  {$\mathcal{T}$};	
\node[rotate = 90] at (3.2, -0.5) {$\underbrace{\hspace{1.3cm}}$};	
\node [] (r nodes) at (3.4, -0.5)  {$r$};

\node [source1] (node 0) at (5.2,0)  {};
		  \node at (5.2,-0.4)  {$\vdots$};
		\node [source1] (node r-1) at (5.2,-1)  {};
		\node [source1] (node 0) at (5.2,-1.4)  {};
		  \node at (5.2,-1.8)  {$\vdots$};
		\node [source1] (node r-1) at (5.2,-2.4)  {};
\node at (5.2,-2.8)  {{$\vdots$}};
\node [source3] (node 0) at (5.2,-3.4)  {};
		  \node at (5.2,-3.8)  {$\vdots$};
		\node [source3] (node r-1) at (5.2,-4.4)  {};
\draw[dashed,rounded corners=5pt] (4.7,-4.7) rectangle (5.7,0.3);

\draw[->] (3.2,-2) -- (4.6,-2);	
\node [] (Transformation) at (3.95,-1.7)  {$\mathcal{T}$};	
\node[rotate = 90] at (5.8, -0.5) {$\underbrace{\hspace{1.3cm}}$};	
\node [] (r nodes) at (6, -0.5)  {$r$};	
\node[rotate = 90] at (5.8, -1.9) {$\underbrace{\hspace{1.3cm}}$};	
\node [] (r nodes) at (6, -1.9)  {$r$};

\draw[->] (6.1,-2) -- (7.5,-2);
\node [] (...) at (7.8, -2)  {$\cdots$};
\draw[->] (8.1,-2) -- (9.5,-2);	

\node [source1] (node 0) at (10.1,0)  {};
		  \node at (10.1,-0.4)  {$\vdots$};
		\node [source1] (node r-1) at (10.1,-1)  {};
		\node [source1] (node 0) at (10.1,-1.4)  {};
		  \node at (10.1,-1.8)  {$\vdots$};
		\node [source1] (node r-1) at (10.1,-2.4)  {};
\node at (10.1,-2.8)  {{$\vdots$}};
\node [source1] (node 0) at (10.1,-3.4)  {};
		  \node at (10.1,-3.8)  {$\vdots$};
		\node [source1] (node r-1) at (10.1,-4.4)  {};
\draw[dashed,rounded corners=5pt] ( 9.6,-4.7) rectangle (10.6,0.3);

\node[rotate = 90] at (10.7, -0.5) {$\underbrace{\hspace{1.3cm}}$};	
\node [] (r nodes) at (10.9, -0.5)  {$r$};	
\node[rotate = 90] at (10.7, -1.9) {$\underbrace{\hspace{1.3cm}}$};	
\node [] (r nodes) at (10.9, -1.9)  {$r$};
\node[rotate = 90] at (10.7, -3.9) {$\underbrace{\hspace{1.3cm}}$};	
\node [] (r nodes) at (10.9, -3.9)  {$r$};	
\node [] (Transformation) at (6.8,-1.7)  {$\mathcal{T}$};	
\node [] (Transformation) at (8.8,-1.7)  {$\mathcal{T}$};	
\node [source2,align=center] (Base node) at (0, .8)  {Base code $\mathcal{Q}_1$};	
\node [source2,align=center] (Base node) at (2.6, .8)  {Base code $\mathcal{Q}_1$};	
\node [source2,align=center] (Base node) at (5.2, .8)  {Base code $\mathcal{Q}_2$};	
\node [source2,align=center] (Base node) at (10.1, .8)  {Final code $\mathcal{Q}_{m+1}$};			
\end{tikzpicture}
\caption{The second application of our generic transformation, where $\mathcal{T}$ denotes our generic transformation and a white (resp. gray) rectangle denotes a storage node without (resp. with) the optimal rebuilding access}
\label{fig app 2}
\end{figure*}

 \begin{table*}[htbp]
\begin{center}
\caption{A $(6,4)$ scalar MDS code $\mathcal{Q}_1$ over $\mathbf{F}_5$, where SN and PN respectively denote systematic node and parity node}\label{Ex Table C1}
\begin{tabular}{|c|c|c|c||c|c|}
\hline
SN $0$  &  SN $1$  &  SN $2$  &  SN $3$ &PN $0$ &PN $1$  \\
\hline\hline
$a_0$  & $b_{0}$& $c_{0}$& $d_{0}$& $a_0+b_{0}+c_0+d_0$& $a_0+2b_{0}+3c_0+4d_0$ \\
\hline
\end{tabular}
\end{center}
\end{table*}

\begin{table*}[htbp]
\begin{center}
\caption{The $(6,4)$ systematic MDS code $\mathcal{Q}_2$, where systematic nodes 0 and 1 are chosen as the target nodes}\label{Ex Table Q2}
\begin{tabular}{|c|c|c|c||c|c|}
\hline
SN $0$  &  SN $1$  &  SN $2$  &  SN $3$ &PN $0$ &PN $1$  \\
\hline\hline
$a_0$  & $b_{0}$& $c_{0}$& $d_{0}$& $a_0+(b_{0}-a_1)+c_0+d_0$& $a_0+2(b_{0}-a_1)+3c_0+4d_0$ \\
\hline
$a_1$  & $b_{1}$& $c_{1}$& $d_{1}$& $(a_1+b_0)+b_{1}+c_1+d_1$& $(a_1+b_0)+2b_{1}+3c_1+4d_1$ \\
\hline
\end{tabular}
\end{center}
\end{table*}

\begin{table*}[htbp]
\begin{center}
\caption{The $(6,4)$ MDS code $\mathcal{Q}_3$ in systematic form, where systematic nodes 2 and 3 are chosen as the target nodes}\label{Ex Table Q3}
\begin{tabular}{|c|c|c|c||c|c|}
\hline
SN $0$  &  SN $1$  &  SN $2$  &  SN $3$ &PN $0$ &PN $1$\\
\hline\hline
$a_0$  & $b_{0}$& $c_{0}$& $d_{0}$& $a_0+(b_{0}-a_1)+c_0+(d_0-c_2)$& $a_0+2(b_{0}-a_1)+3c_0+4(d_0-c_2)$ \\
\hline
$a_1$  & $b_{1}$& $c_{1}$& $d_{1}$& $(a_1+b_0)+b_{1}+c_1+(d_1-c_3)$& $(a_1+b_0)+2b_{1}+3c_1+4(d_1-c_3)$ \\
\hline
$a_2$  & $b_{2}$& $c_{2}$& $d_{2}$& $a_2+(b_{2}-a_3)+(c_2+d_0)+d_2$& $a_2+2(b_{2}-a_3)+3(c_2+d_0)+4d_2$ \\
\hline
$a_3$  & $b_{3}$& $c_{3}$& $d_{3}$& $(a_3+b_2)+b_{3}+(c_3+d_1)+d_3$& $(a_3+b_2)+2b_{3}+3(c_3+d_1)+4d_3$ \\
\hline
\end{tabular}
\end{center}
\end{table*}

Suppose that we choose an $(n,k)$ scalar MDS code, such as a Reed-Solomon code, as the base code $\mathcal{Q}_{1}$.  Let $m=\lceil n/r\rceil$ where $r=n-k$. By applying the transformation $m$ times,  we can get MDS codes $\mathcal{Q}_2$, $\mathcal{Q}_3, \cdots$,  $\mathcal{Q}_{m+1}$. In the $i$-th round transformation, where $i\in[1,m]$, we choose  code $\mathcal{Q}_i$ as the base code, nodes $(i-1)r, (i-1)r+1, \cdots, ir-1$ as the target nodes if $i<m$ and nodes $k, k+1, \cdots, n-1$ otherwise, and denote the resultant code as $\mathcal{Q}_{i+1}$. Obviously, the above procedure eventually gives an MDS code $\mathcal{Q}_{m+1}$ over $\mathbf{F}_{q}$ with $q\ge n$, \textit{i.e.}, the same as the base code, while the sub-packetization level is $r^{\lceil n/r\rceil}$, which matches the lower bound in \cite{tight-bound-on-alpha} and thus optimal except when $r|(n-1)$.
It is worthy noting that the field size of the $(n,k)$ MDS code $\mathcal{Q}_{m+1}$ might be smaller than that of the codes proposed in \cite{Barg2,Sasidharan-Kumar2} when $r\nmid n$, which require the field size
 $q\ge r\lceil{n\over r}\rceil$, as shown in Table \ref{Table comp Ye Kumar}.

Particularly, we can also instead apply the transformation only $t$ times in the above procedure, where $t\in [1,\lceil k/r\rceil]$. This yields an $(n,k)$ MDS code with the optimal rebuilding access for $tr$ nodes, while the sub-packetization level is $r^t$, which is also optimal with respect to the bound for the sub-packetization level of $(n,k)$ MDS codes with the optimal rebuilding access for $tr$ nodes \cite{tight-bound-on-alpha}. Figure \ref{fig app 2} reveals the procedure  of the second application.

In the following, we provide an example of the second application.

\begin{Example}
For a $(6,4)$ systematic  scalar MDS code $\mathcal{Q}_1$ over $\mathbf{F}_5$, the structure of which is listed as in Table \ref{Ex Table C1}.

In the following, we convert the MDS code $\mathcal{Q}_1$ into an MDS code with the optimal rebuilding access for all nodes through three rounds of transformations.  Through three rounds of transformations, we obtain code $\mathcal{Q}_2$, $\mathcal{Q}_3$ and $\mathcal{Q}_4$, which are shown  in Tables \ref{Ex Table Q2}, \ref{Ex Table Q3}, and \ref{Ex Table Q4}, respectively. Especially, in step $2$ of each round, we choose all the permutations as the identity permutation according to Theorem $\ref{Thm repair sys}$ for simplicity. Additionally, in the first, second and third rounds, we respectively choose nodes $0$ and $1$, nodes $2$ and $3$, and the two parity nodes as the target nodes, but only modify the data stored at the parity nodes in each round of transformation, to ensure that the resultant code is of systematic form.

For the code $\mathcal{Q}_4$, it is seen that the code maintains the MDS property. Moreover, systematic nodes $0, 1, 2, 3$, parity nodes $0, 1$ can be respectively repaired by accessing and downloading symbols in rows $\{1,3,5,7\}$, $\{2,4,6,8\}$, $\{1,2,5,6\}$, $\{3,4,7,8\}$, $\{1,2,3,4\}$, $\{5,6,7,8\}$ of Table
\ref{Ex Table Q4}
from each surviving node.
\end{Example}

\begin{table*}[htbp]
\begin{center}
\caption{The $(6,4)$ MDS code $\mathcal{Q}_4$ over $\mathbf{F}_5$ with the sub-packetization level 8, where the two parity nodes are chosen as the target nodes}\label{Ex Table Q4}
\begin{tabular}{|c|c|c|c||c|c|}
\hline
SN $0$  &  SN $1$  &  SN $2$  &  SN $3$ &PN $0$ &PN $1$  \\
\hline\hline
\multirow{2}{*}{$a_0$}  & \multirow{2}{*}{$b_{0}$}& \multirow{2}{*}{$c_{0}$}& \multirow{2}{*}{$d_{0}$}& \multirow{2}{*}{$a_0+(b_{0}-a_1)+c_0+(d_0-c_2)$}& $(a_0+2(b_{0}-a_1)+3c_0+4(d_0-c_2))$\\&&&&&$-(a_4+(b_{4}-a_5)+c_4+(d_4-c_6))$ \\
\hline
\multirow{2}{*}{$a_1$}  & \multirow{2}{*}{$b_{1}$}& \multirow{2}{*}{$c_{1}$}& \multirow{2}{*}{$d_{1}$}& \multirow{2}{*}{$(a_1+b_0)+b_{1}+c_1+(d_1-c_3)$}& $
((a_1+b_0)+2b_{1}+3c_1+4(d_1-c_3))$\\&&&&&$-((a_5+b_4)+b_{5}+c_5+(d_5-c_7))$ \\
\hline
\multirow{2}{*}{$a_2$}  & \multirow{2}{*}{$b_{2}$} & \multirow{2}{*}{$c_{2}$} & \multirow{2}{*}{$d_{2}$} & \multirow{2}{*}{$a_2+(b_{2}-a_3)+(c_2+d_0)+d_2$}& $(a_2+2(b_{2}-a_3)+3(c_2+d_0)+4d_2)$\\
&&&&&$-(a_6+(b_{6}-a_7)+(c_6+d_4)+d_6)$ \\
\hline
\multirow{2}{*}{$a_3$}  & \multirow{2}{*}{$b_{3}$}& \multirow{2}{*}{$c_{3}$}& \multirow{2}{*}{$d_{3}$}& \multirow{2}{*}{$(a_3+b_2)+b_{3}+(c_3+d_1)+d_3$}& $((a_3+b_2)+2b_{3}+3(c_3+d_1)+4d_3)$\\&&&&&$-((a_7+b_6)+b_{7}+(c_7+d_5)+d_7)$ \\
\hline
\multirow{2}{*}{$a_4$}  & \multirow{2}{*}{$b_{4}$}& \multirow{2}{*}{$c_{4}$}& \multirow{2}{*}{$d_{4}$}& $a_4+(b_{4}-a_5)+c_4+(d_4-c_6)$& \multirow{2}{*}{$a_4+2(b_{4}-a_5)+3c_4+4(d_4-c_6)$}\\&&&&$+(a_0+2(b_{0}-a_1)+3c_0+4(d_0-c_2))$ &\\
\hline
\multirow{2}{*}{$a_5$}  & \multirow{2}{*}{$b_{5}$}& \multirow{2}{*}{$c_{5}$}& \multirow{2}{*}{$d_{5}$}& $(a_5+b_4)+b_{5}+c_5+(d_5-c_7)$& \multirow{2}{*}{$(a_5+b_4)+2b_{5}+3c_5+4(d_5-c_7)$}\\&&&&$+((a_1+b_0)+2b_{1}+3c_1+4(d_1-c_3))$& \\
\hline
\multirow{2}{*}{$a_6$}  & \multirow{2}{*}{$b_{6}$} & \multirow{2}{*}{$c_{6}$} & \multirow{2}{*}{$d_{6}$} & $a_6+(b_{6}-a_7)+(c_6+d_4)+d_6$& \multirow{2}{*}{$a_6+2(b_{6}-a_7)+3(c_6+d_4)+4d_6$}\\
&&&&$+(a_2+2(b_{2}-a_3)+3(c_2+d_0)+4d_2)$& \\
\hline
\multirow{2}{*}{$a_7$}  & \multirow{2}{*}{$b_{7}$}& \multirow{2}{*}{$c_{7}$}& \multirow{2}{*}{$d_{7}$}& $(a_7+b_6)+b_{7}+(c_7+d_5)+d_7$& \multirow{2}{*}{$(a_7+b_6)+2b_{7}+3(c_7+d_5)+4d_7$}\\&&&&$+((a_3+b_2)+2b_{3}+3(c_3+d_1)+4d_3)$& \\
\hline
\end{tabular}
\end{center}
\end{table*}

\section{Concluding remarks}

In this paper, we proposed a generic transformation that can be applied on any nonbinary existing MDS code, which produces new MDS codes with some arbitrarily chosen $r$ nodes having the optimal repair bandwidth and the optimal rebuilding access. Furthermore, we provided  two important applications of this transformation to yield MDS codes with the optimal repair property. Given the generic nature of the proposed transformation, we anticipate it can be applied or extended to more cases and then lead to more desired storage codes. In fact, the code construction for delayed parity generation reported in \cite{DPG} is indeed partly inspired by the generic transformation proposed here.

\section*{Acknowledgment}
The authors would like to thank the Associate Editor Chih-Chun Wang and the two anonymous reviewers for their valuable suggestions and comments, which have greatly improved the presentation and quality of this paper.

\begin{IEEEbiographynophoto}{Jie Li} (S'16-M'17) received the B.S. and M.S. degrees in mathematics from the Hubei University, Wuhan, China, in 2009 and 2012, respectively, and received the Ph.D. degree in communication engineering from the Southwest Jiaotong University, Chengdu, China, in 2017.

From Oct. 2015 to Oct. 2016, he was a visiting Ph.D. student in the Department of Electrical Engineering and Computer Science, The University of Tennessee at Knoxville, TN, USA. Currently he is  a postdoctoral fellow at the Department of Mathematics, Hubei University, Wuhan, China. His research interests include coding for distributed storage and sequence design.
Dr. Li was a recipient of the Jack Keil Wolf ISIT Student Paper Award in 2017.
\end{IEEEbiographynophoto}

\begin{IEEEbiographynophoto}{Xiaohu Tang} (M'04)  received the B.S. degree in applied mathematics from
the Northwest Polytechnic University, Xi'an, China, the M.S. degree in applied
mathematics from the Sichuan University, Chengdu, China, and the Ph.D.
degree in electronic engineering from the Southwest Jiaotong University,
Chengdu, China, in 1992, 1995, and 2001 respectively.

From 2003 to 2004, he was a research associate in the Department of Electrical
and Electronic Engineering, Hong Kong University of Science and Technology.
From 2007 to 2008, he was a visiting professor at University of Ulm,
Germany. Since 2001, he has been in the School of Information Science and Technology,
Southwest Jiaotong University, where he is currently a professor. His research
interests include coding theory, network security, distributed storage and information processing for big data.

Dr. Tang was the recipient of the National excellent Doctoral Dissertation
award in 2003 (China), the Humboldt Research Fellowship in 2007
(Germany), and the Outstanding Young Scientist Award by NSFC in 2013
(China). He serves as Associate Editors for several journals including \textit{IEEE
Transactions on Information Theory} and \textit{IEICE Transactions on
Fundamentals}, and served on a number of technical program committees of
conferences.
\end{IEEEbiographynophoto}

\begin{IEEEbiographynophoto}{Chao Tian} (S'00-M'05-SM'12) received the B.E. degree in Electronic
Engineering from Tsinghua University, Beijing, China, in 2000 and the
M.S. and Ph. D. degrees in Electrical and Computer Engineering from
Cornell University, Ithaca, NY in 2003 and 2005, respectively. Dr. Tian
was a postdoctoral researcher at Ecole Polytechnique Federale de Lausanne
(EPFL) from 2005 to 2007, a member of technical staff-research at AT\&T
Labs-Research in New Jersey from 2007 to 2014, and an Associate Professor
in the Department of Electrical Engineering and Computer Science at the
University of Tennessee Knoxville from 2014 to 2017. He joined the Department
of Electrical and Computer Engineering at Texas A\&M University as
an Associate Professor in 2017. His research interests include data storage
systems, multi-user information theory, joint source-channel coding, signal
processing, and compute algorithms.
Dr. Tian received the Liu Memorial Award at Cornell University in 2004,
AT\&T Key Contributor Award in 2010, 2011 and 2013, and 2014 IEEE
ComSoc DSTC Data Storage Best Paper Award. He was an Associate Editor
for the IEEE SIGNAL PROCESSING LETTERS from 2012 to 2014, and is
currently an Editor for the IEEE TRANSACTIONS ON COMMUNICATIONS.
\end{IEEEbiographynophoto}
\end{document}